\documentclass{svproc}
\usepackage{url}

\usepackage{graphicx, latexsym, lscape, amscd, color, epsfig, mathrsfs, tikz, enumerate, hyperref, amsfonts}
\newtheorem{conject}[theorem]{Conjecture}
\begin{document}
\mainmatter              
\title{On Deeply Critical Oriented Cliques}
\titlerunning{On Deeply Critical Oriented Cliques}  
%
\author{Christopher Duffy\inst{1} \and Pavan P D\inst{2} \and
R. B. Sandeep\inst{3} \and Sagnik Sen\inst{2}}

\institute{University of Saskatchewan, Department of Mathematics and Statistics, Saskatoon, CANADA \and Indian Institute of Technology Dharwad, Department of Mathematics, Karnataka, INDIA \and Indian Institute of Technology Dharwad, Department of Computer Science and Engineering, Karnataka, INDIA
}
\maketitle              
\begin{abstract}
In this work we consider arc criticality in colourings of oriented graphs.
We study deeply critical oriented graphs, those graphs for which the removal of any arc results in a decrease of the oriented chromatic number by $2$.
We prove the existence of deeply critical oriented cliques of every odd order $n\geq 9$, closing an open question posed by Borodin et al. [\emph{Journal of Combinatorial Theory, Series B, 81(1):150–155, 2001}].
Additionally, we prove the non-existence of deeply critical oriented cliques among the family of circulant oriented cliques of even order.
\end{abstract}
In 1994, Courcelle~\cite{Co94} defined oriented colouring as part of his seminal work on the monadic second order logic of graphs in which he established the illustrious Courcelle's Theorem~\cite{courcell1995monadic}. 
In the years following, oriented colouring and the oriented chromatic number gained popularity and developed into an independent field of research. 
We refer the reader to Sopena's updated survey~\cite{S16} for a broad overview of the state of the art.

An \emph{oriented graph} $\overrightarrow{G}$ is a directed graph without any directed cycle of length one or two.
That is, it is a directed graph that is irreflexive and anti-symmetric.
We denote the set of vertices and arcs of an oriented graph $\overrightarrow{G}$ by $V(\overrightarrow{G})$ and $A(\overrightarrow{G})$, respectively.
Its underlying simple graph is denoted by $G$.

By generalizing to oriented graphs the interpretation of graph colouring as homomorphism to a complete graph, one arrives at the following definition of oriented graph colouring.

An \emph{oriented $k$-colouring} of an oriented graph $\overrightarrow{G}$ is a function $\phi: V(\overrightarrow{G}) \to \{1,2,\dots k\}$
so that 
\begin{enumerate}[(i)]
	\item $\phi(x) \neq \phi(y)$ for all $xy \in A(\overrightarrow{G})$; and
	\item for all $xy, uv \in A(\overrightarrow{G})$, if $\phi(x) = \phi(v)$, then $\phi(y) \neq \phi(u)$.
\end{enumerate}
The vertices of the target of homomorphism correspond to the colours.
Condition (i) ensures that the target of the homomorphism is irrflexive.
Condition (ii) ensures that the target of the homomorphism is anti-symmetric.

When $y=u$, condition (ii) implies vertices connected by a directed path of length $2$ (i.e., a \emph{$2$-dipath}) must receive different colours.

The \emph{oriented chromatic number} $\chi_o(\overrightarrow{G})$ of an oriented graph $\overrightarrow{G}$ is the minimum $k$ for which $\overrightarrow{G}$ admits an oriented $k$-colouring. 

A major theme in oriented colourings research is the study of analogous versions of graph colouring concepts for oriented graphs. 
In 2004 Klostermeyer and MacGillivray~\cite{KMG04a} generalized the notion of clique to oriented colouring, studying those oriented graphs for which $\chi_o(\overrightarrow{G})= |V(\overrightarrow{G})|$.
This work was continued by Nandi, Sen, and Sopena~\cite{nss16}.
Such oriented graphs are called \emph{absolute oriented cliques} and admit the following classification.
\begin{theorem}[Klostermeyer and MacGillivray]\label{thm:2dipathClassify}
	An oriented graph is an absolute oriented clique if and only if every pair of non-adjacent vertices are connected by a $2$-dipath.
\end{theorem}

In 2001 Borodin, et al.~\cite{borodin2001deeply}, extended the notion of arc criticality for graph colouring to oriented colourings.
Notably they gave examples of oriented graphs for which the removal of any arc decreases the oriented chromatic number by $2$, the maximum possible.
Formally, a \emph{deeply critical oriented graph} is an oriented graph $\overrightarrow{G}$
for which
\[\chi_o(\overrightarrow{G} - xy) = \chi_o(\overrightarrow{G}) - 2\]
for each arc $xy \in A(\overrightarrow{G})$. 

Borodin, et al.~\cite{borodin2001deeply} gave an infinite family of deeply critical oriented graphs that were also absolute oriented cliques.
For convenience, we refer to such an oriented graph as a \emph{deeply critical oriented clique}. 
By way of example, we invite the reader to verify that the directed cycle on $5$ vertices is a deeply critical oriented clique.
\begin{theorem}[Borodin et al.~\cite{borodin2001deeply}]
	\label{th borodin et al}
	There exists a deeply critical oriented clique of order $n$  for every $n=2 \cdot 3^m - 1$, where $m \geq 1$.  
\end{theorem}
In their work Borodin et al.~\cite{borodin2001deeply} speculated the existence of deeply critical oriented cliques of odd order $n$ for all $n \geq 33$ and left it open.
We close this long-standing open problem by proving the following result. 
\begin{theorem}\label{thm:oddDCOC}
Let $n\geq 1$ be an odd integer.
There exists a deeply critical oriented clique of order $n$, if and only if $n \geq 5$, and $n \neq 7$. 
\end{theorem}
We prove this theorem in Section~\ref{sec:thmOddProof}.

A curious aspect of the study of deeply critical oriented cliques is a lack of examples of such oriented graphs of even order, despite intensive computer search.
We conjecture such deeply critical oriented cliques not to exist.
\begin{conject}\label{conj:onlyOdd}
	There exists a deeply critical oriented clique of order $n$, if and only if $n$ is odd, $n \geq 5$, and $n \neq 7$. 
\end{conject}

Let $n$ be an integer and let $S\subseteq \mathbb{Z}_n$ so that for all $k \in \mathbb{Z}_n$ if $k \in S$, 
then $-k \not\in S$. 
Recall that the \emph{oriented circulant graph} $\overrightarrow{C}(n,S)$ is the oriented graph with vertex set $\mathbb{Z}_n$ so that \(ij \in A (\overrightarrow{C}(n,S))\) when $j-i$ is congruent modulo $n$ to an element of $S$.

We provide further evidence towards Conjecture~\ref{conj:onlyOdd} by proving no deeply critical oriented clique appears among the family of oriented circulant graphs of even order.
\begin{theorem}\label{thm:noEvenCirc}
	There does not exist any deeply critical oriented circulant clique of even order. 
\end{theorem}

Our work proceeds as follows.
We prove Theorems~\ref{thm:oddDCOC} and~\ref{thm:noEvenCirc} in Sections~\ref{sec:thmOddProof} and~\ref{sec:noEvenCirc}, respectively.
In this former section we provide a method to construct a deeply critical oriented clique for any odd integer $n\geq 5$, exclusive of $n=7$.
In this latter section we give a full classification of deeply critical oriented circulant cliques.
We provide concluding remarks and suggestions for future work in Section~\ref{sec:Conclusions}.
We refer the reader to~\cite{bondy} for definitions of standard graph theoretic terminology and notation not defined herein.

\section{Proof of Theorem~\ref{thm:oddDCOC}}\label{sec:thmOddProof}
We begin by defining the following terms and notations. 
Let $\overrightarrow{G}$ be an oriented graph. 
An \emph{extending partition} of $\overrightarrow{G}$ is a partition of its set of vertices $V(\overrightarrow{G}) = X_1 \sqcup X_2 \sqcup X_3$ so that 
\begin{enumerate}[(i)]
	\item there is no arc from a vertex of $X_{i+1}$ to a vertex of $X_{i}$, 
	for all $i \in \{1,2,3\}$;
	\item for each $u \in X_i$, there exists a vertex $v \in X_{i+1}$ such that 
	$N^-(v) \cap X_i = \{u\}$, 
	for all $i \in \{1,2,3\}$; and
	\item for each $v \in X_{i+1}$, there exists a vertex $u \in X_{i}$ such that $N^+(u) \cap X_{i+1} = \{v\}$, 
	for all $i \in \{1,2,3\}$,
\end{enumerate}
where addition in indices is taken modulo $3$.

We say an oriented graph is \textit{extendable} when it admits an extending partition.
For such graphs we define the following supergraphs.
Let $\overrightarrow{G}$ be an oriented graph with  extending partition $V(\overrightarrow{G}) = X_1 \sqcup X_2 \sqcup X_3$. 
The \emph{$6$-extension} of  $\overrightarrow{G}$  is the graph $\overrightarrow{G}_6$ constructed from $\overrightarrow{G}$ as follows (see Figure~\ref{fig:6extend}): 

\begin{itemize}
	\item Include six new vertices $x_1^{-}$, $x_1^{+}$, $x_2^{-}$,
	$x_2^{+}$, $x_3^{-}$, $x_3^{+}$  to the graph $\overrightarrow{G}$, and add the arcs $x_1^{-}x_2^{+}$, $x_1^{+}x_2^{-}$, $x_2^{-}x_3^{+}$, $x_2^{+}x_3^{-}$, $x_3^{-}x_1^{+}$, and $x_3^{+}x_1^{-}$. 
	\item Add all the arcs of the form $x_1^{-}x$ and $xx_1^{+}$ for all $x \in X_1$. 
	\item Add all the arcs of the form $x_2^{-}x$ and $xx_2^{+}$ for all $x \in X_2$.
	\item Add all the arcs of the form $x_3^{-}x$ and $xx_3^{+}$ for all $x \in X_3$.
\end{itemize}

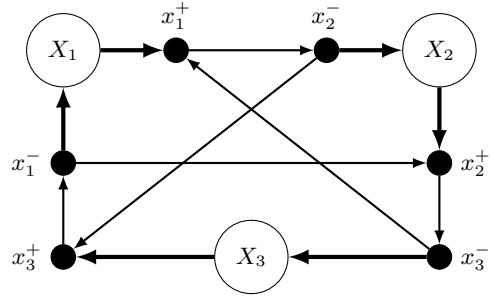
\begin{figure}[htbp]
	\centering
	\begin{tikzpicture}
		\tikzset{vertex/.style = {shape=circle,draw,minimum size=1em, fill=black}}
		\tikzset{edge/.style = {->,> = latex}}
		\node[vertex,label=left:$x_1^{-}$] (a1) at  (0.5,0) {};
		\node[vertex, label = above:$x_1^{+}$] (a2) at  (2,1.5) {};
		\node[vertex, label = above:$x_2^{-}$] (b1) at  (4,1.5) {};
		\node[vertex, label = right:$x_2^{+}$] (b2) at  (5.5,0) {};
		\node[vertex, label = right:$x_3^{-}$] (c1) at  (5.5,-1.25) {};
		\node[vertex, label = left:$x_3^{+}$] (c2) at  (0.5,-1.25) {};
		
		\tikzset{vertex/.style = {shape=circle,draw,minimum size=3em}}
		\node[vertex] (A) at  (0.5,1.5) {$X_1$};
		\node[vertex] (B) at  (5.5,1.5) {$X_2$};
		\node[vertex] (C) at  (3,-1.25) {$X_3$};
		\draw[edge, ultra thick] (a1) to (A);
		\draw[edge, ultra thick] (A) to (a2);
		\draw[edge, thick] (a2) to (b1);
		\draw[edge, ultra thick] (b1) to (B);
		\draw[edge, ultra thick] (B) to (b2);
		\draw[edge, thick] (b2) to (c1);
		\draw[edge, ultra thick] (c1) to (C);
		\draw[edge, ultra thick] (C) to (c2);
		\draw[edge, thick] (c2) to (a1);
		\draw[edge, thick] (c1) to (a2);
		\draw[edge, thick] (b1) to (c2);
		\draw[edge, thick] (a1) to (b2);
	\end{tikzpicture}
	\caption{Construction of $\protect{\overrightarrow{G}_6}$. Thickened arcs indicate the existence of all possible arcs between the vertex and the set $X_i$. Arcs between sets of vertices are not shown.}
	\label{fig:6extend}
\end{figure}

Following the definition $6$-extension, we define two further extensions of \(\overrightarrow{G}\), which arise as induced subgraphs of  $\overrightarrow{G}_6$.
The \emph{$4$-extension} of  
$\overrightarrow{G}$ is the graph obtained from  $\overrightarrow{G}_6$ by 
deleting the vertices $x_1^{+}$ and $x_2^{-}$ from $\overrightarrow{G}_6$. 
The \emph{$2$-extension} of  $\overrightarrow{G}$  is the graph obtained from  $\overrightarrow{G}_6$ deleting the vertices $x_1^{-}, x_2^{+}, x_3^{-}$ and $x_3^{+}$ from $\overrightarrow{G}_6$.

\begin{lemma}\label{lem:ExtensionClique}
	Let $\overrightarrow{G}$ be an extendable deeply critical oriented clique.
	The  $2$-extension, $4$-extension, and the $6$-extension of $\overrightarrow{G}$ are deeply critical oriented cliques. 
\end{lemma}
\begin{proof}
Let $\overrightarrow{G}$ be a deeply critical oriented clique having a extending partition $V(\overrightarrow{G}) = X_1 \sqcup X_2 \sqcup X_3$. 
Let $\overrightarrow{G}_6$ be the 6-extension of $\overrightarrow{G}$. We first verify that $\overrightarrow{G}_6$ is an oriented clique.

Note that, due to symmetry of the construction,  it is enough to verify that 
$x_1^{+}$ and $x_1^{-}$ see every other vertex of $\overrightarrow{G}_6$ , where by ``a vertex $u$ \textit{sees} a vertex $v$" we mean that $u$ and $v$ are either adjacent or connected by a 2-dipath.

\medskip

Observe that $x_1^{-}$ sees each vertex of $X_1$ directly, each vertex $v$ of $X_2$ through some $u$ of $X_1$ due to property (iii) of the definition of extending partition, and each vertex of $X_3$ through $x_3^{+}$. 
Moreover, $x_1^{-}$ sees $x_1^{+}$ through vertices of $X_1$, $x_2^{-}$ through 
$x_3^{+}$, $x_2^{+}$ directly, $x_3^{-}$ through $x_2^{+}$, and 
$x_3^{+}$ directly. 

Similarly, observe that $x_1^{+}$ sees each vertex of $X_1$ directly, 
each vertex of $X_2$ through $x_2^{-}$, and
each vertex $u$ of $X_3$ through some $v$ of $X_1$ due to 
property (ii) of the definition of extending partition. 
Moreover, $x_1^{+}$ see $x_1^{-}$ through vertices of $X_1$, $x_2^{-}$ directly, $x_2^{+}$ through $x_3^{-}$,  $x_3^{-}$ directly, and 
$x_3^{+}$ through $x_2^{-}$.

The above arguments are valid in showing that the induced subgraphs $\overrightarrow{G}_2$ and $\overrightarrow{G}_4$ are also oriented cliques.

\medskip

Since, $\overrightarrow{G}$ is deeply critical, and the addition of the new vertices does not create any new adjacencies or $2$-dipaths between the vertices of $X_1$, $X_2$, and $X_3$, to verify whether $\overrightarrow{G}_6$ is deeply critical we need to only check whether removing the newly added arcs decreases the chromatic number of $\overrightarrow{G}_6$ by $2$. Furthermore, due to the symmetry of $\overrightarrow{G}_6$, it is enough to check 
for the arcs $x_1^{-}u$ and $ux_1^{+}$ for some $u \in X_1$, $x_1^{+}x_2^{-}$, and $x_1^{-}x_2^{+}$.

Thus, first let us consider the oriented graph
$\overrightarrow{G}_6 - x_1^{-}u$ for any $u \in X_1$.    Note that, due to property (ii), there exists a $v \in X_2$ which is neither adjacent nor connected by a $2$-dipath with $x_1^{-}$ any more. 
Hence, it is possible to assign the same color to $x_1^{-}$ and $v$, and on the other hand, it is possible to assign the same color to $x_3^{+}$ and $u$. If we assign all distinct colors to rest of the vertices of $\overrightarrow{G}_6 - x_1^{-}u$, then what we get is an oriented coloring. 
Thus, $\chi_o(\overrightarrow{G}_6 - x_1^{-}u) = |\overrightarrow{G}_6|-2$. Similarly, one can show that 
$\chi_o(\overrightarrow{G}_6 - u x_1^{+}) = |\overrightarrow{G}_6|-2$ for any $u \in X_1$.
Therefore, we are done.

Next let us consider the oriented graph 
$\overrightarrow{G}_6 - x_1^{-} x_2^{+}$. In this case, we can assign a particular color to the vertices $x_1^{-}, x_3^{-}$ and another particular color to the vertices $x_2^{+}, x_3^{+}$. The rest of the vertices can be assigned distinct colors. This gives us 
$\chi_o(\overrightarrow{G}_6 - x_1^{-} x_2^{+}) = |\overrightarrow{G}_6|-2$. 
Similarly, one can show that 
$\chi_o(\overrightarrow{G}_6 - x_1^{+} x_2^{-}) = |\overrightarrow{G}_6|-2$.
Hence, $\overrightarrow{G}_6$ is a deeply critical oriented clique.

The same arguments are valid in showing that $\overrightarrow{G}_2$ and $\overrightarrow{G}_4$ are also deeply critical oriented cliques. \qed
\end{proof}

Lemma~\ref{lem:ExtensionClique} implies that given an oriented deeply critical oriented clique on $n$ vertices, one may construct deeply critical oriented cliques on $n+2$, $n+4$ and $n+6$ vertices.
We note, however that computer search yields many examples of deeply critical oriented cliques that do not arise as an extension of a smaller deeply critical oriented clique.
Figure~\ref{fig:DC9_2} gives such an example.
Curiously, though generated by computer search, the oriented graph in Figure~\ref{fig:DC9_2} does arise as a $6$-extension of a directed three cycle.

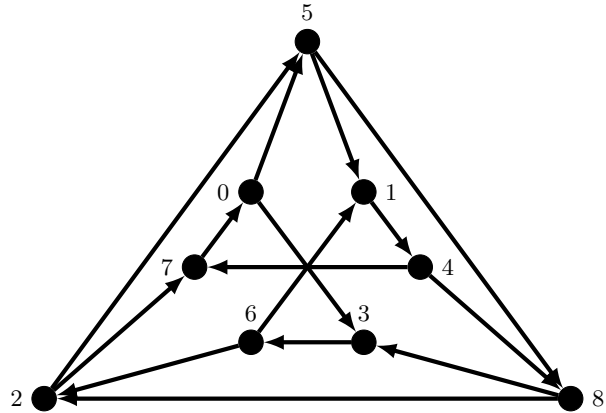
\begin{figure}[htbp]
	\centering
	\begin{tikzpicture}
		\tikzset{vertex/.style = {shape=circle,draw,minimum size=1em, fill=black}}
		\tikzset{edge/.style = {->,> = latex}}
		\node[vertex, label =$5$ ] (a) at  (3.75,5.25) {};
		\node[vertex, label =right:$8$ ] (b) at  (7.25,0.5) {};
		\node[vertex, label =left:$2$ ] (c) at  (0.25,0.5) {};
		\node[vertex, label =left:$7$ ] (d) at  (2.25,2.25) {};
		\node[vertex, label =left:$0$ ] (e) at  (3,3.25) {};
		\node[vertex, label =right:$1$ ] (f) at  (4.5,3.25) {};
		\node[vertex, label =right:$4$ ] (g) at  (5.25,2.25) {};
		\node[vertex, label =$3$ ] (h) at  (4.5,1.25) {};
		\node[vertex, label =$6$ ] (i) at  (3,1.25) {};
		\draw[edge, ultra thick] (a) to (b);
		\draw[edge, ultra thick] (b) to (c);
		\draw[edge, ultra thick] (c) to (a);
		\draw[edge, ultra thick] (c) to (d);
		\draw[edge, ultra thick] (i) to (c);
		\draw[edge, ultra thick] (b) to (h);
		\draw[edge, ultra thick] (g) to (b);
		\draw[edge, ultra thick] (e) to (a);
		\draw[edge, ultra thick] (a) to (f);
		\draw[edge, ultra thick] (d) to (e);
		\draw[edge, ultra thick] (f) to (g);
		\draw[edge, ultra thick] (h) to (i);
		\draw[edge, ultra thick] (e) to (h);
		\draw[edge, ultra thick] (g) to (d);
		\draw[edge, ultra thick] (i) to (f);
	\end{tikzpicture}
	\caption{A deeply critical oriented clique on 9 vertices.}
	\label{fig:DC9_2}
\end{figure}
Given $\overrightarrow{G}$, a deeply critical oriented clique with extending partition $V(\overrightarrow{G}) = X_1 \sqcup X_2 \sqcup X_3$, one may verify $V(\overrightarrow{G}_6) = X'_1 \sqcup X'_2 \sqcup X'_3$ where $X_i' = X_i \cup \{x_i^{-}, x_i^{+}\}$ for all $i \in \{1,2,3\}$ is an extending partition of $\overrightarrow{G}_6$.
\begin{lemma}\label{lem:6extendable}
	The $6$-extension of an extendable deeply critical oriented clique is extendable.
\end{lemma}

\begin{proof}
Let $\overrightarrow{G}$ be a deeply critical oriented clique having a extending partition $V(\overrightarrow{G}) = X_1 \sqcup X_2 \sqcup X_3$. 
We claim that $\overrightarrow{G}_6$, which is also a deeply critical oriented clique due to Lemma~\ref{lem:ExtensionClique}, admits a extending partition 
$V(\overrightarrow{G}) = X'_1 \sqcup X'_2 \sqcup X'_3$, where we have 
 $X_i' = X_i \cup \{x_i^{-}, x_i^{+}\}$ for all $i \in \{1,2,3\}$. 
 
 As the complete description of $\overrightarrow{G}_6$ is available, it is straight forward to verify the claim. \qed
\end{proof}

With these two lemmas in place, we provide a  proof of Theorem~\ref{thm:oddDCOC}. 

\begin{proof}[Theorem~\ref{thm:oddDCOC}]
	The directed cycle on $5$ vertices is a deeply critical oriented clique.
	By computer search there is no deeply critical oriented clique on $7$ vertices.
	
	The oriented graph given in Figure~\ref{fig:DC9_2} is a  deeply critical oriented clique with $9$ vertices.
	This oriented graph admits the following extending partition:
	 \(X_1 = \{6, 2, 7\}, X_2 = \{1, 5, 0\}, X_3 = \{4, 8, 3\}\).
	The result now follows inductively from the following observation:
	By Lemmas~\ref{lem:ExtensionClique} and~\ref{lem:6extendable}, if $\overrightarrow{H}$ is an extendable deeply critical oriented clique with $n$ vertices, then there exists deeply critical oriented clique on $n+2$ and $n+4$ vertices, and an extendable deeply critical oriented clique on $n+6$ vertices. \qed
\end{proof}

\section{Proof of Theorem~\ref{thm:noEvenCirc}}\label{sec:noEvenCirc}
We provide a proof of Theorem~\ref{thm:noEvenCirc} by first giving a full classification of deeply critical oriented circulant cliques.
\begin{lemma}
	\label{lem:char}
	The circulant graph $\overrightarrow{C}(n,S)$ is a deeply critical oriented clique if and only if for every $k \in \mathbb{Z}_n$
	\begin{enumerate}[(a)]
		\item there exists $x,y \in S \cup \{0\}$ so that \(k \equiv x +y \pmod n\) or \(k \equiv -(x+y)\pmod n\); and
		\item if $k$ is even and $\frac{k}{2} \in S$ then the only way to express $k$ as in (a) is by taking $x=y=\frac{k}{2}$, and writing $k \equiv (x + y) \pmod n$.
	\end{enumerate}
\end{lemma}

\begin{proof}
First let us prove the ``if'' part. Therefore, assume that $\overrightarrow{C}(n,S)$ satisfies the given conditions. The indices mentioned hereafter are assumed to be taken $\pmod n$.

Note that, for any $x \in S$, there is an arc from $v_{i-x}$ to $v_{i}$ and  from $v_{i}$ to $v_{i+x}$. Thus, for any non-zero values of $x, y \in S$, 
There exists a $2$-dipath from $v_{i-(x+y)}$ to $v_{i}$ and from $v_{i}$ to $v_{i+(x+y)}$.
Thus condition (a) of the statement implies that, for any two vertices 
of $\overrightarrow{C}(n,S)$ there is either an arc or a $2$-dipath. 
Hence, $\overrightarrow{C}(n,S)$ is an oriented absolute clique due to 
Theorem~\ref{thm:2dipathClassify}.

Next we will show that $\overrightarrow{C}(n,S)$ is deeply critical. 
To do so, we need to show that removal of any arc $v_iv_{i+x}$ will 
decrease the oriented chromatic number by exactly $2$. As we know~\cite{borodin2001deeply} that removal of an arc can decrease the oriented chromatic number by at most $2$, it will be enough to show that there exists an oriented $(n-2)$-coloring of 
$[\overrightarrow{C}(n,S) - v_iv_{i+x}]$ for each $x \in S$. 
Moreover, as $\overrightarrow{C}(n,S)$ is vertex transitive, it is enough to consider $i = 0$, in particular for our proof. 

Thus, let $\overrightarrow{C}'$ be the oriented graph obtained by
deleting the arc $v_0v_x$ from $\overrightarrow{C}(n,S)$, for some $x \in S$. 
Therefore, due to condition (b) of the statement, 
we know that there is no arc or $2$-dipath connecting the pairs of vertices 
$(v_{-x}, v_{x})$ and $(v_{0}, v_{2x})$.  Note that, the oriented coloring of 
$\overrightarrow{C}'$ given by assigning color $1$ to $v_{-x}, v_x$, 
color $2$ to $v_0, v_{2x}$, and $(n-4)$ distinct colors to the other vertices is 
an oriented $(n-2)$-coloring of $\overrightarrow{C}'$. Thus we are done with the ``if'' part of the proof.

\medskip

Now we will prove the ``only if'' part of the proof. Assume that $\overrightarrow{C}(n,S)$ is a deeply critical oriented clique. 

First note that condition (a) of the statement is trivial for $k = 0$. 
Suppose, if condition (a) is not satisfied for some integer $k \neq 0$, 
then we can say that $v_0$ and $v_{k}$ are neither adjacent nor connected by a $2$-dipath. This contradicts the fact that $\overrightarrow{C}(n,S)$ is an oriented absolute clique. 

Next suppose that condition (b) of the statement is false for some even $k \in S$. Observe that, surely $k \equiv (\frac{k}{2}+\frac{k}{2}) \pmod n$ is one way to express $k$. 
However, as condition (b) is false, there is another way of expressing 
$k$ according to condition (a), that is $k = x + y \pmod n$ or $k = -(x + y) \pmod n$.

Let $\overrightarrow{C}'$ be the oriented graph obtained by
deleting the arc $v_0v_{\frac{k}{2}}$ from $\overrightarrow{C}(n,S)$. 
Note that, all the adjacent vertices of $v_0$ in $\overrightarrow{C}(n,S)$, 
except  $v_{\frac{k}{2}}$, remains adjacent in $\overrightarrow{C}'$. 
Suppose, $v_{l}$ is a vertex which was not adjacent to $v_0$ in $\overrightarrow{C}(n,S)$. As we know that $\overrightarrow{C}(n,S)$ is an oriented absolute clique, due to Theorem~\ref{thm:2dipathClassify}, there must be a $2$-dipath connecting $v_0$ to $v_l$. 
Observe that, there is yet another $2$-dipath $v_0v_{y}v_{x+y}$ connecting $v_0$ 
and $v_l$ unless $x=y=\frac{l}{2}$. 
If the $2$-dipath was from 
$v_0$, then it must have been of the form $v_0v_{x}v_{x+y}$, where $(x+y) = l$. The case if the $2$-dipath is from $v_l$ to $v_0$ is similar. 

Therefore, even after removing the arc 
$v_0v_{\frac{k}{2}}$, all vertices except maybe $v_{\frac{k}{2}}$, will 
remain adjacent or connected by a $2$-dipath with $v_0$.

Similarly, we can show that even after removing the arc 
$v_0v_{\frac{k}{2}}$, all vertices except maybe $v_0$, will 
remain adjacent or connected by a $2$-dipath with $v_{\frac{k}{2}}$. 

As any pair of vertices not including $v_0$ or $v_{\frac{k}{2}}$ have not used the arc $v_0v_{\frac{k}{2}}$ for being adjacent or connected by a $2$-dipath 
in $\overrightarrow{C}(n,S)$, notice that those pairs remain adjacent or connected by a $2$-dipath in $\overrightarrow{C}'$. 

Thus, as every pairs of vertices, except maybe $(v_0, v_{\frac{k}{2}})$ are either adjacent or connected by a $2$-dipath in $\overrightarrow{C}'$, 
the oriented chromatic number of $\overrightarrow{C}'$ must be at least $(n-1)$. This contradicts that fact that $\overrightarrow{C}(n,S)$ is a deeply critical oriented graph. Thus, the converse is also proved. \qed
\end{proof}

Part (a) ensures that $\overrightarrow{C}(n,S)$ is an absolute oriented clique.
As circulant oriented graphs are vertex transitive we need only verify that (a) is equivalent to vertex $0$ being either adjacent or connected by a $2$-dipath to every vertex in $\overrightarrow{C}(n,S)$.
When $k \equiv x +y \pmod n$ or $k \equiv -(x+y) \pmod n$ and $x,y \neq 0$, there is a $2$-dipath, in some direction, between $0$ and $k$.
Otherwise if $x=0$ or $y=0$, then $0$ and $k$ are adjacent.

Part (b) ensures $\overrightarrow{C}(n,S)$ is deeply critical.
We appeal to vertex transitivity and use the stated condition to verify that for all vertices $k$ not adjacent to $0$ there is exactly one $2$-dipath between $0$ and $k$. 

\begin{proof}[Proof of Theorem~\ref{thm:noEvenCirc}]
Let $\overrightarrow{C}(n,S)$ be an oriented circulant graph so that $\overrightarrow{C}(n,S)$ is an absolute clique and $n$ is even.
As $n$ is even we have $\frac{n}{2} \equiv - \frac{n}{2} \pmod n$.
Therefore $\frac{n}{2} \not\in S$.
Subsequently vertices $0$ and $\frac{n}{2}$ are not adjacent in $\overrightarrow{C}(n,S)$.

Since $\overrightarrow{C}(n,S)$ is an absolute clique, by Theorem~\ref{thm:2dipathClassify}
there is a $2$-dipath connecting vertices $0$ and \(\frac{n}{2}\).
Thus, there exists $x,y \in S$ satisfying $(x+y) \equiv \frac{n}{2} \pmod n$ or $-(x+y) \equiv \frac{n}{2} \pmod n$.
Therefore, $2x \equiv x+x \equiv -(y+y) \pmod n$ or $-2x \equiv -x-x \equiv y+y \pmod n$, violating part (b) of 
Lemma~\ref{lem:char}. 
Therefore, $\overrightarrow{C}(n,S)$ is not a deeply critical oriented clique.  \qed
\end{proof}

\section{Conclusions and Outlook}\label{sec:Conclusions}
Work in~\cite{borodin2001deeply}  and herein provide examples of infinite families of deeply critical oriented cliques.
These constructions and extensive computer search have yielded no examples of deeply critical oriented cliques of even order.
These observations together with the result of Theorem~\ref{thm:noEvenCirc} lend support the statement of Conjecture~\ref{conj:onlyOdd}.

Our computer search has yielded surprising insight into the density of deeply critical oriented cliques among the family of absolute oriented cliques.
We identified 9917 examples of previously unknown sporadic deeply critical oriented cliques on up to 17 vertices.
In addition to these examples, our search of oriented circulants found 28 examples of previously unknown deeply critical circulant oriented cliques on up to 49 vertices.
A classification of odd orders for which there exists a deeply critical circulant oriented clique remains open.

A result of Erd\"os~\cite{E63} implies that asymptotically almost surely, every oriented graph is an absolute oriented clique.
Though attempts to extend this result to deeply critical oriented cliques of odd order have so far been unsuccessful, it is possible that an analogous statement is true for deeply critical oriented cliques.

\bibliographystyle{abbrv}
\bibliography{References}

\end{document}